\documentclass{article}

\usepackage{url}

\usepackage{amsmath}
\usepackage{amssymb}
\usepackage{amsthm}
\usepackage{verbatim}
\usepackage{graphicx}

\newcommand{\confey}{\includegraphics[width=4mm,height=4mm]{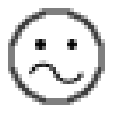}}
\newcommand{\smiley}{\includegraphics[width=4mm,height=4mm]{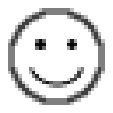}}
\newcommand{\frowney}{\includegraphics[width=4mm,height=4mm]{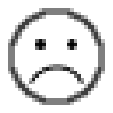}}
\newcommand{\neutrey}{\includegraphics[width=4mm,height=4mm]{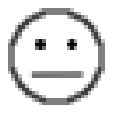}}

\newtheorem{thm}{Theorem}[subsection]
\newtheorem{claim}[thm]{Claim}
\newtheorem{cor}[thm]{Corollary}

\newtheorem{con}[thm]{Conjecture}
\newtheorem{neutrey_con}[thm]{\neutrey~~Conjecture}
\newtheorem{smiley_con}[thm]{\smiley~~Conjecture}
\newtheorem{smiley_confey_con}[thm]{$\mathrm{\smiley}\mathrm{\confey}$~~Conjecture}

\newtheorem{frowney_question}[thm]{\frowney~~Question}
\newtheorem{question}[thm]{Question}
\newtheorem{example}[thm]{Example}

\theoremstyle{remark}
\newtheorem{rem}[thm]{Remark}

\theoremstyle{definition}
\newtheorem{defn}[thm]{Definition}

\DeclareMathOperator{\rank}{\operatorname{rank}}
\DeclareMathOperator{\CREW}{\operatorname{CREW}}
\DeclareMathOperator{\poly}{\operatorname{poly}}

\begin{document}

\title{Variations on the Sensitivity Conjecture}
\author{
{\bf Pooya Hatami} \and
{\bf Raghav Kulkarni} \and
{\bf Denis Pankratov} \and {\small} \\
\\[-8mm]
{University of Chicago, \it{Department of Computer Science}} \\
{\small e-mail: pooya, raghav, pankratov@cs.uchicago.edu}
}


\maketitle

\begin{abstract}

We present a selection of known as well as new variants of the
Sensitivity Conjecture and  point out some weaker versions that are
also open.

\end{abstract}


\setcounter{thm}{0}
\numberwithin{thm}{section}

\section{The Sensitivity Conjecture}
Let $f : \{0, 1\}^n \rightarrow \{0, 1\}$ be a Boolean function.  Let
$e_i \in \{0, 1\}^n$ denote an $n$-bit Boolean string whose
$i^\mathrm{th}$ bit is $1$ and the rest of the bits are $0.$ On an
input $x \in \{0, 1\}^n,$ the $i^\mathrm{th}$ bit is said to be {\em
  sensitive} for $f$ if $f(x \oplus e_i) \neq f(x),$ i.e., flipping
the $i^\mathrm{th}$ bit results in flipping the output of $f.$ The
sensitivity of $f$ on input $x,$ denoted by $s(f, x)$, is the number of
bits that are sensitive for $f$ on input $x.$

\begin{defn}
  The {\em sensitivity} of a Boolean function $f,$ denoted by $s(f)$,
  is the maximum value of $s(f, x)$ over all choices of $x$.
\end{defn}

Study of sensitivity of Boolean functions originated from Cook and
Dwork~\cite{Cook82} and Reischuk~\cite{Reischuk82}. They showed an
$\Omega(\log s(f))$ lower bound on the number of steps required to
compute a Boolean function $f$ on a CREW PRAM. A CREW PRAM,
abbreviated from Consecutive Read Exclusive Write Parallel RAM, is a
collection of synchronized processors computing in parallel with
access to a shared memory with no write conflicts. The minimum number
of steps required to compute a function $f$ on a CREW PRAM is denoted
by $\CREW(f)$. After Cook, Dwork and Reischuk introduced sensitivity,
Nisan~\cite{Nisan89} found a way to modify the definition of
sensitivity to characterize $\CREW(f)$ exactly. He introduced a
related notion called {\em block sensitivity}. A block $B$ is a subset
of $[n] = \{1, 2, \ldots, n \}$. Let $e_B \in \{0, 1\}^n$ denote the
characteristic vector of $B,$ i.e., the $i^\mathrm{th}$ bit of $e_B$
is $1$ if $i \in B$ and $0$ otherwise.  We say that a block $B$ is
sensitive for $f$ on $x$ if $f(x \oplus e_B) \neq f(x)$. The block
sensitivity of $f$ on $x$, denoted by $bs(f,x)$, is the maximum number
of {\em pairwise disjoint} sensitive blocks of $f$ on $x$.

\begin{defn}
The {\em block sensitivity} of a Boolean function $f$, denoted by $bs(f)$,
is the maximum possible value of $bs(f,x)$ over all choices of $x$.
\end{defn}

Obviously,  for every Boolean function $f$,
\[ s(f) \le bs(f). \]

Nisan's influential result~\cite{Nisan89} states that $\CREW(f) =
\Theta(\log bs(f))$ for every Boolean function $f$. Block sensitivity
turned out to be polynomially related to a number of other complexity
measures (see Section~\ref{MeasuresEquivSec}); however, to this day it
remains unknown whether block sensitivity is bounded above by a
polynomial in sensitivity. The following conjecture, known as the
Sensitivity Conjecture, is due to Nisan and Szegedy~\cite{Nisan92}.

\begin{con}[{\bf Sensitivity Conjecture}, Nisan and Szegedy]
\label{BSCon}
For every Boolean function $f$,
\[bs(f) \le \poly(s(f)).\]
\end{con}


The rest of the paper is organized as follows. In
Section~\ref{MeasuresEquivSec} we describe complexity measures of
Boolean functions polynomially related to block sensitivity. In
Section~\ref{ProgressSec} we review progress on the Sensitivity
Conjecture. In Section~\ref{OtherSec} we present alternative
formulations of the Sensitivity Conjecture and point out weaker
versions that are also open. Along the way we encounter important
examples of Boolean functions. We present these functions in
Section~\ref{ExamplesSec}.

\section{Measures Related to Block Sensitivity}
\label{MeasuresEquivSec}

Block sensitivity is polynomially related to several other complexity
measures of Boolean functions, which we describe in this section.

A {\em deterministic decision tree} on $n$ variables $x_1, \ldots,
x_n$ is a rooted binary tree, whose internal nodes are labeled with
variables, and the leaves are labeled $0$ or $1$. Edges are also
labeled $0$ or $1$. To evaluate such a tree on input $x$, start at the
root and query the corresponding variable, then move to the next node
along the edge labeled with the outcome of the query. Repeat until a
leaf is reached, at which point the label of the leaf is declared to
be the output of the evaluation. A decision tree computes a Boolean
function $f$ if it agrees with $f$ on all inputs. 
\begin{defn}
  The {\em deterministic decision tree complexity} of a Boolean
  function $f$, denoted by $D(f)$, is the depth of a minimum-depth
  decision tree that computes $f$.
\end{defn}

One way to extend the deterministic decision tree model is to add
randomness to the computation. In this extended model, each node $v$
has an associated bias $p_v \in [0,1]$. Evaluation proceeds as before,
except when deciding which edge to follow after querying $x_i$ at node
$v$, we follow an edge corresponding to the outcome of the query with
probability $p_v$ and the other edge with probability $1-p_v$. 

\begin{defn}
The {\em bounded-error randomized decision tree complexity} of a Boolean
function $f$, denoted by $R_2(f)$, is the depth of a minimum-depth
randomized decision tree computing $f$ with probability at least $2/3$
for all $x \in \{0,1\}^n$.
\end{defn}

A {\em certificate} of a Boolean function $f$ on input $x$, is a
subset $S \subset [n]$, such that $(\forall y \in \{0,1\}^n) (x|_S =
y|_S \Rightarrow f(x) = f(y))$. The {\em certificate complexity} of a
Boolean function $f$ on input $x$, denoted by $C(f,x)$, is the minimum
size of a certificate of $f$ on $x$.


\begin{defn} 
  The {\em certificate complexity} of a Boolean function $f$,
  also known as {\em non-deterministic decision tree complexity} and
  denoted by $C(f)$, is the maximum of $C(f, x)$ over all choices of
  $x.$
\end{defn}

\begin{defn}
A polynomial $p : \mathbb{R}^n \rightarrow \mathbb{R}$ {\em represents} $f$
if \[(\forall x \in \{0,1\}^n)(p(x) = f(x)).\] The {\em degree} of a
Boolean function $f$, denoted by $\deg(f)$,  is the degree of the unique
multilinear polynomial that represents $f.$
\end{defn}

\begin{defn}
A polynomial $p : \mathbb{R}^n \rightarrow \mathbb{R}$ {\em approximately
represents} $f$ if
\[(\forall x \in \{0, 1\}^n)(|p(x) - f(x)| < 1/3).\]
The {\em approximate degree} of a Boolean function $f$, denoted by
$\widetilde \deg (f)$, is the minimum degree of a polynomial that
approximately represents $f.$
\end{defn}

We denote the {\em quantum decision tree complexity with bounded
  error} of a Boolean function $f$ by $Q_2(f)$. Discussion of quantum
complexity is outside the scope of this note. For an introduction to
quantum complexity see a survey by de Wolf~\cite{Wolf2002}.

\begin{defn}
Complexity measures $A$ and $B$ are {\em polynomially
  related} if 
\[ (\forall f) \left[ A(f) \le \poly(B(f)) \;
  \mathrm{and} \; B(f) \le \poly(A(f)) \right] .\]
\end{defn}

\begin{thm} 
\label{PolyEquiv}
The following complexity measures of Boolean functions are all
polynomially related:

\begin{tabular}{ccccccc}
$bs(f)$, & $D(f)$,  & $R_{2}(f)$, &
$C(f)$, & $deg(f)$, &$\widetilde{deg} (f)$, &$Q_2(f).$
\end{tabular}
\end{thm}

Table~\ref{tabel} presents a quick summary of the known polynomial
relations between complexity measures that play a prominent role in
this note. An entry from the table shows the smallest known exponent
of a polynomial in the corresponding measure from the column that
gives an upper bound on the corresponding measure from the row, as
well as the exponent of the biggest known gap between two measures. An
entry also contains references to papers, where the result can be
found. 
References of the form [*] indicate that the result is immediate from
the definitions of complexity measures. For example, entry
3~\cite{Midrijanis04}($\log_3 6$~\cite{Nisan95}) in the second row and
third column means that $D(f)= O(\deg(f)^3)$ (see
~\cite{Midrijanis04}) and there is a Boolean function $f$, for which
$D(f) = \Omega(\deg(f)^{\log_3 6})$ (see ~\cite{Nisan95}).
For a thorough treatment of polynomial relation between various
complexity measures of Boolean functions (including variants of
quantum query complexity) see a survey by Buhrman and de
Wolf~\cite{Buhrman02}.

\begin{table}
\begin{center}
\begin{tabular}{|c|| c | c | c | c |}
\hline
 & $bs(f)$ & $D(f)$ & $\deg(f)$ & $C(f)$ \\
\hline
\hline
$bs(f)$ & 1(1) & 1~[*](1~[*]) & 2~\cite{Nisan92}($\log_3 6$~\cite{Nisan95}\footnotemark[3]) & 1~[*](1~[*]) \\
\hline
$D(f)$ & 3~\cite{Beals01,Nisan89}(2~[]\footnotemark[4]) & 1(1) & 3~\cite{Midrijanis04}($\log_3 6$~\cite{Nisan95}\footnotemark[3]) & 2~\cite{Beals01}\footnotemark[2](2~[]\footnotemark[4]) \\
\hline
$\deg(f)$ & 3~\cite{Beals01,Nisan89}(2~[]\footnotemark[4]) & 1~[*](1~[*]) & 1(1) & 2~\cite{Beals01}(2~[]\footnotemark[4]) \\
\hline
$C(f)$ & 2~\cite{Nisan89}($\log_4 5$~\cite{Bublitz86,Aaronson03}\footnotemark[1]) & 1~[*](1~[*]) & 3~\cite{Midrijanis04}($\log_3 6$~\cite{Nisan95}\footnotemark[3]) & 1(1) \\
\hline

\end{tabular}
\end{center}
\caption{Known polynomial relations between various complexity measures. An entry in the table shows the polynomial upper bound on the measure from a row in terms of a measure from a column and the biggest known gap between two measures. The references to the papers, where the corresponding results can be found, are given in square brackets.}
\label{tabel}
\end{table}

\footnotetext[1]{The construction appeared in ~\cite{Bublitz86} before
  the notion of block sensitivity was introduced. The analysis of
  $C(f)$ and $bs(f)$ of the example appears in ~\cite{Aaronson03}.}
\footnotetext[2]{The result is due
  to~\cite{Blum87,Hartmanis87,Tardos89}.}  \footnotetext[3]{The
  example is due to Kushilevitz and appears in footnote 1 on p. 560 of
  the Nisan-Wigderson paper~\cite{Nisan95}. See
  Example~\ref{Kushilevitz} in Section~\ref{ExamplesSec} of this
  paper.}  \footnotetext[4]{ These gaps are demonstrated by a commonly
  known AND-of-ORs function, see Example~\ref{AndOfOr} in
  Section~\ref{ExamplesSec} of this paper.}

Using Theorem~\ref{PolyEquiv}, one immediately obtains many equivalent
formulations of the ``sensitivity versus block sensitivity''
conjecture. The purpose of this note is to point out some nontrivial
variations on this conjecture that, to our knowledge, have not been
stated explicitly in the literature. We also propose several weaker
versions of the Sensitivity Conjecture, which might provide starting
points.

We introduce the following pictorial notation to indicate relations
between the statements appearing in this note and the Sensitivity
Conjecture.

\noindent \smiley - a consequence of the Sensitivity Conjecture.

\noindent \frowney - implies the Sensitivity Conjecture, but the
reverse implication is not known. These might be good candidates for
refutation.

\noindent \neutrey - equivalent to the Sensitivity Conjecture.

\noindent \confey - conditionally equivalent to the Sensitivity
Conjecture.









\section{Progress on the Sensitivity Conjecture}
\label{ProgressSec}

The progress on the Sensitivity Conjecture has been limited.  Simon
\cite{simon} proved that for any Boolean function that depends on all
$n$ variables, sensitivity is at least $\frac{1}{2} \log n - \frac{1}{2} \log \log n +
\frac{1}{2}$. An immediate corollary is that for any Boolean function $f$,
$bs(f) = O\left(s(f)4^{s(f)}\right)$. Kenyon and Kutin \cite{kk}
proved that sensitivity is polynomially related to $\ell$-block
sensitivity for any constant $\ell$ ($\ell$-block sensitivity
considers only the sensitive blocks of size at most $\ell).$ The best
known upper bound on block sensitivity in terms of sensitivity is
exponential and appears in the work of Kenyon and Kutin~\cite{kk} on
$\ell$-block sensitivity:
\[bs(f) \le \left(2/\sqrt{2\pi}\right) e^{s(f)} \sqrt{s(f)}.\]

In the 80s, Rubinstein~\cite{rub} exhibited a function with
sensitivity $\Theta(\sqrt{n})$ and block sensitivity $\Theta(n)$ (see
Example~\ref{RubEx} in Section~\ref{ExamplesSec}). Gaps between
sensitivity and some other complexity measures are surveyed by Buhrman
and de Wolf \cite{Buhrman02}.

In the light of Rubinstein's example, the best possible upper bound on
block sensitivity in terms of sensitivity could be quadratic. Nisan
and Szegedy \cite{Nisan92} asked the following question:

\begin{question}
Is $bs(f) = O(s(f)^2)$ for every Boolean function $f$?
\end{question}

A Boolean function $f : \{0,1\}^n \rightarrow \{0,1\}$ is {\em
  invariant under a permutation} $\sigma : [n] \rightarrow [n]$, if
for any string $x$, $f(x_1, \ldots, x_n) = f(x_{\sigma(1)},
\ldots, x_{\sigma(n)})$. The set of all permutations, under which $f$
is invariant, forms a group, called the {\em invariance group} of $f$.
A Boolean function is said to be {\em transitive} if its invariance
group $\Gamma$ is {\em transitive}, i.e., for each $i,j \in [n]$ there
is a permutation $\sigma \in \Gamma$ such that $\sigma(i) = j$. 

Tur\'an \cite{turan} proved that any property of $n$-vertex graphs
(viewed as a Boolean function on $n \choose 2$ variables) has
sensitivity $\Omega(n)$. Tur\'an asked if every transitive function on
$n$ variables has sensitivity at least $\Omega(\sqrt n)$. Chakraborty
\cite{chak} answered this question in the negative by constructing a
transitive function with sensitivity $\Theta(n^{1/3})$ and block
sensitivity $\Theta(n^{2/3})$ (see Example~\ref{ChakrabortyEx} in
Section~\ref{ExamplesSec}). We propose the following modification of
Tur\'an's question:

\begin{question}
\label{QuestionTrans}
If $f : \{0, 1\}^n \rightarrow \{0, 1\}$ is transitive and $f({\bf 0})
\neq f({\bf 1}),$ is then $s(f) = \Omega( \sqrt n )$?
\end{question}

An example of a different behavior of transitive Boolean functions
with the property $f({\bf 0}) \neq f({\bf 1})$ is due to Rivest and
Vuillemin~\cite{rv}. They proved that if $n$ is a prime power, and
$f({\bf 0})\neq f({\bf 1})$, then $D(f)=n$. From their proof it can be
immediately inferred that in fact $\deg(f)=n$.  A conjecture appearing
in the Gotsman-Linial paper~\cite{Gotsman92} states that $\deg(f) =
O(s(f)^2)$, which, if true, would answer Question~\ref{QuestionTrans}
positively for $n$ that are prime powers.


\section{Sensitivity vs Other Complexity Measures}
\label{OtherSec}
Unlike the complexity measures mentioned in
Section~\ref{MeasuresEquivSec}, the complexity measures in this
section are not polynomially related to block sensitivity and yet
proving a polynomial relation of these measures to sensitivity turns
out to be equivalent to proving a polynomial relation between block
sensitivity and sensitivity, i.e., the Sensitivity Conjecture itself.

Let $F(x, y)$ be a Boolean function.  Consider a setting in which
Alice has a Boolean string $x$ and Bob has a Boolean string $y$, and
their goal is to compute the value of $F(x, y)$ by communicating as
few bits as possible. Alice and Bob agree on a communication protocol
beforehand. Having received inputs, they communicate in accordance
with the protocol. At the end of the communication one of the parties
declares the value of the function $F$. The cost of the protocol is
the number of bits exchanged on the worst-case input.

\begin{defn}
The {\em deterministic communication complexity} of $F$, denoted by
$DC(F)$, is the cost of an optimal communication protocol computing
$F$.
\end{defn}

For more information on communication complexity see
\cite{Kushilevitz97}. Given a Boolean function on $n$ variables, we
will typically consider $F(x, y) = f(x \circ y)$ where $\circ$ is
bitwise $\wedge, \vee$ or $\oplus.$

\numberwithin{thm}{subsection}

\subsection{Log-rank vs Sensitivity}
\label{LogrankForm}

For a Boolean function $F$ of two arguments, $\rank(F)$ denotes the
rank of the corresponding matrix $M_{x,y} = F(x,y)$ over $\mathbb{R}$.

In this section we present some implications of a recent result by
Sherstov. It appears as Theorem 6.4 in~\cite{Sherstov2010Equiv}.

\begin{thm}[Sherstov]
\label{SherstovDegThm}
For every Boolean function $f$,
\[ \max\{\log \rank(f(x \wedge y)), \log \rank (f(x \vee y)) \} =
\Omega(\deg(f)).\]
\end{thm}

\begin{neutrey_con}
\label{LogrankCon}
For every Boolean function $f$, 
\[\log \rank[f(x \wedge y)] \le \poly(s(f)).\]
\end{neutrey_con}

\begin{thm} 
\label{LogrankThm}
Sensitivity Conjecture $\iff$ Conjecture~\ref{LogrankCon}.
\end{thm}
\begin{proof} \ 

  $\Rightarrow$ Beals {\it et al.}~\cite{Beals01} showed that $D(f) \le
  C(f)bs(f)$ and Nisan~\cite{Nisan89} proved that $C(f) \le
  bs(f)^2$. Therefore, we get  $D(f) \le bs(f)^3$ as a simple corollary.
  It is easy to see that $DC(f(x \wedge y)) \le 2D(f)$.
  Finally, by a classical result in communication complexity due to
  Mehlhorn and Schmidt~\cite{Mehlhorn82}, $(\forall F)(\log
  \rank(F(x,y))) \le DC(F)))$.


$\Leftarrow$ For a Boolean function $f$, define $g(x) = f( \lnot
  x)$. Clearly, $s(g) = s(f)$ and $\log \rank(g(x \wedge y)) = \log
  \rank(f(x \vee y))$.  Applying the hypothesis to both $g$ and $f$,
  we get that
  \[ \max\{\log \rank(f(x \vee y)) ,\log \rank(f(x \wedge
  y)) \} \le \poly(s(f)). \]
  It follows that $\deg(f) \le \poly(s(f))$ by
  Theorem~\ref{SherstovDegThm}. This completes the proof, since $bs(f)
  \le 2 \deg(f)^2$~\cite{Nisan92}.
\end{proof}

\begin{neutrey_con}
\label{ParityLogrankCon}
For every Boolean function $f$, 
\[\log \rank[f(x \oplus y)] \le \poly(s(f)).\]
\end{neutrey_con}

\begin{cor}
\label{LogrankRem}
Sensitivity Conjecture $\iff$ Conjecture~\ref{ParityLogrankCon}.
\end{cor}

\begin{proof} \ 

$\Rightarrow$ Similar to the same direction in Theorem~\ref{LogrankThm}.

$\Leftarrow$ For a Boolean function $f$, define $F(x,y) = f(x \wedge
  y)$. It is easy to check that $s(f) \leq s(F) \leq 2 s(f),$ and also
  that 
\[\rank(f(x \wedge y)) \le \rank(F(x \oplus x',
 y \oplus y')).\] 
The result now follows from Theorem~\ref{LogrankThm}.
\end{proof}

\begin{defn}
The {\em sign-rank} of Boolean function $F$ of two arguments, denoted by
$\rank_\pm$, is defined as
\[\rank_\pm (F) = \min_M \left\{ \rank\left( M \right) ~ \left| ~ \left( \forall x,y\right) ~\left((-1)^{F(x,y)} M_{x,y}  > 0 \right) \right\}\right. .\]
\end{defn}
The notion of sign-rank was introduced by Paturi and
Simon~\cite{Paturi86} to give a characterization of the
unbounded error probabilistic communication complexity.


Since $\rank_\pm (F) \le \rank(F)$ for every $F$, we propose a
possibly weaker version of Conjecture~\ref{ParityLogrankCon} stated
for the sign-rank.

\begin{smiley_con}
\label{PMRankCon}
For every Boolean function $f$, 
\[\log \rank_\pm(f(x \oplus y)) \le \poly(s(f)).\]
\end{smiley_con}

\begin{question}
Does Conjecture~\ref{PMRankCon} imply
Conjecture~\ref{ParityLogrankCon}? I.e., is Conjecture~\ref{PMRankCon}
equivalent to the Sensitivity Conjecture?
\end{question}

\subsection{Parity Decision Trees}
{\em Parity decision trees} are similar to decision trees; the
difference is that instead of querying only one variable at a time,
one may query the sum modulo $2$ of an arbitrary subset of variables
(see~\cite{Zhang2010} for a brief introduction to parity decision
trees). The parity decision tree complexity of a Boolean function $f$
is denoted by $D_\oplus(f)$. Obviously, $D_\oplus(f) \le D(f)$. In
this section, we explore the relationship between $D_\oplus$ and
sensitivity. Note that parity decision trees are strictly more
powerful than decision trees. For instance, parity of $n$ bits
requires a decision tree of depth $n$ whereas a parity decision tree
of depth $1$ suffices.
\begin{neutrey_con}
\label{ParityCon}
For every Boolean function $f$, $D_\oplus (f) \le \poly(s(f))$.
\end{neutrey_con}

This seemingly weaker conjecture is actually equivalent to the
Sensitivity Conjecture.

\begin{thm}
\label{ParityThm}
Sensitivity Conjecture $\iff$ Conjecture~\ref{ParityCon}.
\end{thm}

\begin{proof} \ 

$\Rightarrow$  $D_\oplus (f) \le D(f) \le bs(f)^3$.

$\Leftarrow$ Since $\log \rank(f(x
  \oplus y)) \le DC(f(x \oplus y)) \le 2 D_\oplus (f),$ 
  the proof follows from Corollary~\ref{LogrankRem}.
\end{proof}

Next we present a quadratic gap between $D_\oplus$ and
sensitivity. Consider the Boolean function $h({\bf
  x})=\bigwedge_{i=1}^{2^k} \bigvee_{j=1}^{2^k} x_{ij},$ on $n=2^{2k}$
variables. Clearly, $s(h) = 2^k = \sqrt{n}$. To see that $D_\oplus(h)
= n$, consider the mod $2$ degree of $h$ defined as:

\begin{defn}
The \emph{mod $2$ degree} of a Boolean funtion $f$, denoted by $\deg_\oplus (f)$, is the degree
of the unique multilinear polynomial over $\mathbb{F}_2$ (the field of two elements)
that represents $f.$
\end{defn}
Observe that the OR function (and consequently the AND function) has full mod $2$ degree. It follows that $h$ has full mod $2$ degree, which shows that $D_\oplus(h) = n$, since for any Boolean function $f$, $\deg_\oplus(f) \le D_\oplus(f)$.



Similar to the question stated in the survey by Buhrman and de Wolf~\cite{Buhrman02} whether $D(f) \le O(bs(f)^2)$,
we ask the following:

\begin{question}
\label{ParityQ}
Is $D_\oplus (f) = O(bs(f)^2)$?
\end{question}
\begin{rem}
  A positive answer to the above question would imply that \\
  $\deg(f)~\leq~bs(f)^2,$ improving the current best known bound
  $\deg(f) \leq bs(f)^3$ (see~\cite{Beals01}).
\end{rem}



\subsection{Analytic Setting}
\label{AnalSecs}

In the previous sections, we considered Boolean functions from
$\{0,1\}^n$ to $\{0,1\}$. For the purpose of studying the Fourier
spectrum of Boolean functions, it is convenient to use range $\{+1,
-1\}$, replacing $0$ with $+1$ and $1$ with $-1$. This operation
preserves the complexity measures up to an additive constant. For a
brief introduction to Fourier Analysis on the Boolean cube, see, for
instance, the survey by de Wolf~\cite{Wolf2008}.

\begin{defn}
For $S \subseteq [n]$, the {\em character} $\chi_S$ is defined as
\[ \chi_S(x) = (-1)^{\sum_{i \in S} x_i}. \]
\end{defn}

\begin{defn}
The {\em Fourier coefficient} of $f$ corresponding
to $S$ is defined as
 \[ \widehat{f}(S) := \mathop{\mathbb E}_{x \in \{0, 1\}^n} \left[f(x) \chi_S(x) \right].\]
\end{defn}

\begin{neutrey_con}
\label{MinFourierCon}
For every Boolean function $f$, 
\[\min_{S: \widehat{f}(S) \neq 0}
| \widehat{f}(S) | \geq 2^{- \poly(s(f))}.\]
\end{neutrey_con}

\begin{thm}
\label{MinFourierThm}
Sensitivity Conjecture $\iff$ Conjecture~\ref{MinFourierCon}.
\end{thm}

\begin{proof}\ 

  $\Rightarrow$ It is easy to see that if $f$ has a decision tree of
  depth $d$ then all non-zero Fourier coefficients are integer
  multiples of $2^{-d}$. The result follows from $D(f) \le bs(f)^3$,
  as stated in the proof of Theorem~\ref{LogrankThm}.

$\Leftarrow$ Let $\alpha = \min_{S:\widehat{f}(S) \neq 0}
  |\widehat{f}(S)|.$ Since $\sum_S \widehat{f}(S)^2 = 1$ (Parseval's
  Identity), the number of non-zero Fourier coefficients is at most
  $\alpha^{-2}$. Consider matrix $M$ with entries $M_{x,y} = f(x \oplus
  y)$. It is easy to check that for each $S \subset [n]$, the vector
  $(\chi_S(y))_{y \in \{0,1\}^n}$ is an eigenvector to $M$ with a
  corresponding eigenvalue $2^n \widehat{f}(S)$. Since the $\chi_S$
  form an orthogonal set of vectors, the $2^n \widehat{f}(S)$ are all
  the eigenvalues of $M$.

  Hence, $\alpha \geq 2^{- \poly(s(f))} $ implies that rank of $f(x
  \oplus y)$ is at most $2^{\poly(s(f))}.$ The proof is complete by
  Corollary~\ref{LogrankRem}.
\end{proof}

The following consequence of the Sensitivity Conjecture appears to be
open.

\begin{smiley_confey_con}
\label{SumFourierCon}
For every Boolean function $f$, \[ \sum_S | \widehat{f}(S) | \le
2^{\poly(s(f))}. \]
\end{smiley_confey_con}

\begin{defn}
  Let $F(x, y)$ be a Boolean function. Suppose Alice has a Boolean
  string $x$ and Bob has a Boolean string $y.$ The {\em bounded-error
    randomized communication complexity with} {\em shared randomness}
  of $F$, denoted by $RC_2(F)$, is the least cost of a randomized
  protocol that computes $F$ correctly with probability at least $2/3$
  on every input, when Alice and Bob are given the same random bits.
\end{defn}


Next we prove that Conjecture~\ref{SumFourierCon} is equivalent to the
Sensitivity Conjecture under the following variant of the Log-rank
Conjecture due to Grolmusz~\cite{Grolmusz97}.

\begin{con}
[Grolmusz] 
Let $F: \{0,1\}^{m+n} \rightarrow \{-1, 1\}.$
Suppose Alice has $x \in \{0,1\}^m$ and Bob has $y \in \{0,1\}^n,$ then:
\[ RC_2(F(x, y)) \leq \poly (\log \sum_{S \subseteq [m + n]} |\widehat{F}(S)|) .\]
\label{GrolCon}
\end{con}

To prove the equivalence we will need the following result by
Sherstov(see ~\cite{Sherstov2010Equiv}, Theorem 5.1).
\begin{thm}
[Sherstov]
\label{SherstovThm}
Let $F_1(x, y) := f(x \wedge y)$ and $F_2(x, y) := f(x \vee y),$ then
\[\max\{
RC_2(F_1), RC_2(F_2)= \Omega(bs(f)^{1/4}).\]
\label{RC}
\end{thm}

\begin{thm} Conjecture~\ref{GrolCon} $\Rightarrow$ (Sensitivity
  Conjecture $\iff$ Conjecture~\ref{SumFourierCon}).
\end{thm}
\begin{proof}

Assume Conjecture~\ref{GrolCon}. Now, we want to prove that
Sensitivity Conjecture $\iff$ Conjecture~\ref{SumFourierCon}.\\

$\Rightarrow$ $ \begin{array}[t]{llll}
1 & = & \sum_S \widehat{f}(S)^2 & \text{by Parseval's Identity} \\
& \ge & (\min_S |\widehat{f}(S)|) (\# \{S \; |\; \widehat{f}(S) \neq 0\}) & \ \\
& \ge & 2^{-\poly(s(f))} (\sum_S |\widehat{f}(S)|) & \text{by Theorem~\ref{MinFourierThm} }\\
\end{array}
$

$\Leftarrow$
Consider two Boolean functions $F_1$ and $F_2$ on $2n$ variables
defined as in Theorem~\ref{SherstovThm}. It is easy to check that $s(f)
\leq s(F_1)$ and $s(F_2) \leq 2 s(f).$ Applying
Conjecture~\ref{SumFourierCon} to both $F_1$ and $F_2,$ we get: $\log
\sum |\widehat{F_1}(S)| \leq \poly (s(f))$ and $\log \sum
|\widehat{F_2}(S)| \leq \poly(s(f)).$ Now $bs(f) \leq \poly(s(f))$
follows from Theorem~\ref{RC} assuming Conjecture~\ref{GrolCon}.

\end{proof}

\subsection{Shi's Characterization of Sensitivity}
In this section we present some applications of Shi's work~\cite{Shi02}, which contains an interesting characterization of the sensitivity
of Boolean functions. 

A polynomial representing a Boolean function $f: \{0,1\}^n \rightarrow
\{0, 1\}$ (see Section~\ref{MeasuresEquivSec}) provides a multilinear
extension of $f$ from $\mathbb{R}^n$ to $\mathbb{R}$, which (abusing
notation) we denote by the same letter $f$.

Let $\ell=({\bf a}, {\bf b})$ denote the line segment in $[0,1]^n$
that starts at point $\bf a$ and ends at point $\bf b.$

\begin{defn}
The \emph{linear restriction} of $f$ on $\ell=({\bf a},{\bf b})$, $f_\ell:[0,1]\rightarrow
 \mathbb{R}$, is defined as
\[ f_\ell(t) := f((1-t){\bf a}+t{\bf b}),\;\; \forall t \in [0,1]. \]
\end{defn}
 
\noindent Denote the supremum norm of a function $g : [0,1] \rightarrow \mathbb{R}$ by
$||g||_\infty = \sup_{t \in [0,1]} |g(t)|$. Let $g'$ denote the first derivative of $g.$
\begin{thm} [Shi~\cite{Shi02}]
\label{ShiThm}
For every Boolean function $f$, $s(f) = \sup_\ell ||f^\prime_\ell||_\infty$.
\end{thm}
\begin{proof}
It is easy to check that it suffices to consider the lines that join two points of
the Boolean cube.

For $x \in [0, 1]^n,$ let $x^{(i,1)}$ ($x^{(i,0)}$) denote a vector whose $i^{th}$ coordinate
is $1$ ($0$) and the other coordinates match with those of $x.$ 
Let ${\bf a}, {\bf b} \in \{0, 1\}^n$ and $\ell = ({\bf a}, {\bf b})$ be the line joining $a$ and $b.$
\[ f'_\ell(t) = \mathop{\sum}_{i=1}^n (b_i - a_i) \cdot \frac{ \partial f}{\partial x_i}( (1-t){\bf a} + t {\bf b}). \]
Since $f$ is multilinear, we have:
\[ \frac{\partial f}{\partial x_i}( x ) = f(x^{(i, 1)}) - f(x^{(i, 0)}) .\]
Thus we have:
\begin{equation}\label{EllDeriv}
|f'_\ell(t) | \leq  \mathop{\mathbb{E}}_{{\bf p} \in D_t}\left[\mathop{\sum}_{i=1}^n \left|f({\bf p}^{(i,1)}) - f({\bf p}^{(i,0)})\right|\right],
\end{equation}
where $D_t$ denotes the following probability distribution on Boolean cube: for each
$k,$ $\Pr(p_k = 1) = (1-t)a_k + b_k.$ 
Notice that the right hand side of (\ref{EllDeriv}) is at most $s(f).$

For the other direction let ${\bf a} \in \{0, 1\}^n$ and 
$\bf b$ be obtained from $\bf a$ by flipping each bit.
It is easy to check that:
\[ f'_\ell(0) = \mathop{\sum} |f({\bf a} \oplus e_i) - f(a)| = s(f, {\bf a}).\]
Choosing a vector ${\bf a}$ with maximum sensitivity completes the proof.
\end{proof}
Combining Theorem~\ref{ShiThm} with Conjecture~\ref{MinFourierCon}
puts the original ``sensitivity versus block sensitivity'' problem into an
 analytic setting.

\begin{defn}
The {\em approximate degree of linear restrictions} of a Boolean function $f$ is defined as follows:
\[ \overline \deg (f) = \max_\ell \min \left\{ \deg(g) \left| \:  g \in \mathbb{R}[t], ||f_\ell-g||_\infty \le 1/6  \right. \right\}. \]
\end{defn}

\begin{thm}[Shi~\cite{Shi02}] 
\label{ShiEquiv}
The complexity measures $\overline \deg(f)$ and $s(f)$ are polynomially
related.
\end{thm}

Observe that, unlike all previous equivalence results,
Theorem~\ref{ShiEquiv} gives a complexity measure
polynomially related to $s(f)$ rather than $bs(f)$.  It follows that the
following conjecture is equivalent to the Sensitivity Conjecture.

\begin{neutrey_con}[Shi~\cite{Shi02}]
For every Boolean function $f$, \[\widetilde \deg (f) \le \poly \left(\overline \deg (f)\right).\]
\end{neutrey_con}

\subsection{Subgraphs of the $n$-cube}
\label{SubgraphSec}

Let $Q_n$ denote the $n$-cube graph, i.e., $V(Q_n) = \{0,1\}^n$ and
two vertices are adjacent if the corresponding vectors differ in
exactly one position. Denote the maximum degree of graph $G$ by
$\Delta(G)$. For a subgraph $H$ of a graph $G$ define
$\Gamma(H)=\max\{\Delta(H),\Delta(G-H)\}$. Gotsman and
Linial~\cite{Gotsman92} proved the following remarkable equivalence.

\begin{thm}[Gotsman and Linial~\cite{Gotsman92}]
\label{GraphThm}
The following are equivalent for any monotone function $h: \mathbb{N}
\rightarrow \mathbb{R}$:
\begin{description}
\item[A] For any induced subgraph $G$ of $Q_n$ with $|V(G)|\neq 2^{n-1}$ we have $\Gamma(G) \geq h(n)$.
\item[B] For any Boolean function $f$ we have $s(f) \ge h(\deg(f))$.
\end{description}
\end{thm}

\begin{proof} \

  Statement {\bf B} is equivalent to the following:
\begin{description}
\item[B'] For any Boolean function $f$ with $\deg(f) = n$ we have $s(f)
  \ge h(n)$.
\end{description}
Clearly, {\bf B} implies {\bf B'}. To prove the reverse implication,
let $f$ be a Boolean function of degree $d$. Fix a monomial of degree
$d$ of the representing polynomial of $f$. Without loss of generality we
may assume the monomial is $x_1 \cdots x_d$. Define $g(x_1, \ldots,
x_d) := f(x_1, \ldots, x_d, 0, \ldots, 0)$. Then, $s(f) \ge s(g) \ge
h(d)$, as desired.

{\bf A} $\Rightarrow$ {\bf B'} We prove the contrapositive. Given a
Boolean function $f$ with $s(f) < h(n)$, consider an induced subgraph
$G$ of $Q_n$ with
\[V(G)=\{x\in \{0,1\}^n \; |\; f(x)p(x)=+1\},\]
 where $p(x)=(-1)^{\sum x_i}$ is the parity function (as in
 Section~\ref{AnalSecs}, we take the range of Boolean functions to be
 $\{+1, -1\}$).  Observe that $\widehat{f}(I)=\widehat{fp}([n]-I)$ for
 any subset $I \subseteq [n]$. Hence, $\widehat{fp}(\emptyset) =
 \widehat{f}([n]) \neq 0$, since $\deg(f) = n$. Straight from the
 defintion of Fourier coefficients, $\widehat{f}(\emptyset) =
 \mathbb{E}_x[f(x)]$, so $|V(G)| \neq 2^{n-1}$. Furthermore, $s(fp, x)
 = n - s(f,x)$ and $\deg_G(x) = n - s(fp, x) = s(f,x)$. Thus,
 $\Gamma(G) < h(n)$.

{\bf A} $\Leftarrow$ {\bf B'} Observe that the steps in the proof of
{\bf A} $\Rightarrow$ {\bf B'} are reversible.
\end{proof}

The proof of Theorem~\ref{GraphThm} translates a Boolean function with
a polynomial gap between degree and sensitivity into a graph with the
same polynomial gap between $\Gamma$ and $n$, and vice versa. For
example, observe that Rubinstein's function (see Example~\ref{RubEx}
in Section~\ref{ExamplesSec}) has sensitivity $\Theta(n)$ and full
degree, which can be easily verified by a direct computation of
$\widehat{f}([n])$. Therefore, Rubinstein's function can be used to
obtain a graph $G$ with the surprising property $\Gamma(G) =
\Theta(\sqrt{n})$. Chung {\em et al.}~\cite{Chung88} independently
constructed a graph $G$ with $\Gamma(G)<\sqrt{n}+1$. Their example can
be also obtained from Theorem~\ref{GraphThm} by applying the reduction
in the proof of {\bf A} $\Rightarrow $ {\bf B'} to the AND-of-ORs
function (see Example~\ref{AndOfOr} in Section~\ref{ExamplesSec}), but
note that the Gotsman-Linial theorem was not available at the time
when Chung et. al. gave their construction.

It immediately follows that the following
conjecture is equivalent to the Sensitivity Conjecture, by taking $h$ to
be an inverse polynomial in Theorem~\ref{GraphThm}.

\begin{neutrey_con}
\label{GraphCon}
For any induced subgraph $G$ of $Q_n$ such that \\ $|V(G)| \neq
2^{n-1}$ we have $\Gamma(G) \ge \poly^{-1}(n)$.
\end{neutrey_con}


\subsection{Two-colorings of Integer Lattices}

Two points $a, b \in \mathbb{Z}^d$ are called neighbors if $||a-b||_2
= 1$. A two-coloring $C$ of $\mathbb{Z}^d$ with colors red and blue is
non-trivial if the origin is colored red, and there is a point colored
blue on each of the coordinate axes. Sensitivity of a point $a \in
\mathbb{Z}^d$ under coloring $C$, denoted by $S(a,C)$, is the number
of neighbors of $a$ that are colored differently from $a$. 
\begin{defn}
The {\em sensitivity of a coloring} is defined by $S(C) = \max_a S(a,C)$.
\end{defn}
Aaronson~\cite{Aaronson2010} stated the following question, a
positive answer to which would imply the Sensitivity Conjecture. For completeness, we also present a reduction.

\begin{frowney_question}[Aaronson]
\label{LatticeCon}
Does every non-trivial coloring of $\mathbb{Z}^d$ have sensitivity
at least $d^{\Omega(1)}$?
\end{frowney_question}



\begin{claim}[Aaronson]
  A positive answer to Question~\ref{LatticeCon} implies that
  Sensitivity Conjecture.
\end{claim}
\begin{proof} Given a Boolean function $f$ on $n$ variables, let $x$
  be an input, on which $f$ achieves the highest block sensitivity $b
  = bs(f)$.  Let $S_1, \ldots, S_b$ be pairwise disjoint sensitive
  blocks of $f$ on $x$, and let $R = [n] - (\bigcup_i S_i)$. Let
  $\gamma_i : \mathbb{Z} \rightarrow \{0,1\}^{|S_i|}$ represent a Gray
  code with $\gamma_i(0) = x|_{S_i}$. Consider the following mapping
  $\phi : \mathbb{Z}^b \rightarrow \{0,1\}^n$: a point $a \in
  \mathbb{Z}^b$ is mapped to a Boolean string $y \in \{0, 1\}^n$ with
  $y|_{S_i} = \gamma_i(a_i)$ and $y|_R = x|_R$. Finally, obtain
  coloring $C$ of $\mathbb{Z}^b$ by composing $f$ with
  $\phi$. Clearly, $C$ is non-trivial and $s(C) \le 2 s(f)$, hence
  $bs(f) = b \le \poly(s(C)) \le \poly(s(f))$.
\end{proof}

\setcounter{thm}{0}
\numberwithin{thm}{section}
\section{Some Boolean Functions}
\label{ExamplesSec}

In this section we present some interesting examples of Boolean
functions. They provide lower or upper bounds for various complexity
measures, and some of them appear in more than one context.

The following function was exhibited by Rubinstein~\cite{rub}. It was
discussed in Section~\ref{ProgressSec} and Section~\ref{SubgraphSec}
of this note.
\begin{example}[Rubinstein's function]
\label{RubEx}
Rubinstein's function is defined on $n = k^2$ variables,
which are divided into $k$ blocks with $k$ variables each. The value
of the function is $1$ if there is at least one block with exactly two
consecutive $1$s in it, and it is $0$ otherwise.
\end{example}

Block sensitivity of Rubinstein's function on $k^2$ variables is
$\Theta(k^2)$ (hence, certificate complexity and decision tree
complexity is also $\Theta(k^2)$) and sensitivity is $\Theta(k)$. It
has full degree as can be verified by a direct computation of Fourier
coefficient of the set $[k^2]$.

The following folklore example was discussed in
Section~\ref{MeasuresEquivSec} and Section~\ref{SubgraphSec} of this
paper.
\begin{example}[AND-of-ORs function]
\label{AndOfOr}
AND-of-ORs function is defined on $k$ blocks of $k$ variables each:
\[f(x_{11}, \ldots, x_{kk}) = \bigwedge_{i=1}^k \bigvee_{j=1}^k x_{ij}. \]
\end{example}

The block sensitivity and sensitivity of AND-of-ORs function on $k^2$
variables is $k$. AND-of-ORs has full degree and hence its decision tree
complexity is also $k^2$. The certificate complexity of AND-of-ORs function
is $k$.

\begin{defn}
  For a Boolean function $f:\{0,1\}^m \rightarrow \{0,1\}$ and a
  Boolean function $g:\{0,1\}^n \rightarrow \{0,1\}$, we define the {\em
    composition} function $f \diamond g$ on $mn$ variables as follows:
  \[ f \diamond g(x_{11}, \ldots, x_{mn}) = f\big(g(x_{11}, \ldots, x_{1n}),
  \ldots, g(x_{m1}, \ldots, x_{mn})\big).\]
\end{defn}

Kushilevitz exhibited a function $f$ that provides the largest gap in
the exponent of a polynomial in $\deg(f)$ that gives an upper bound on
$bs(f)$. Never published by Kushilevitz, the function appears in
footnote 1 of the Nisan-Wigderson paper~\cite{Nisan95}. It was
discussed in Section~\ref{MeasuresEquivSec} of this paper.
\begin{example}[Kushilevitz's function]
\label{Kushilevitz}
Define an auxiliary function $h$ on $6$ variables:
\[
\begin{array}{lll}
h(z_1, \ldots, z_6) &= &\sum_i z_i - \sum_{ij} z_i z_j + z_1 z_3 z_4 + z_1 z_2 z_5 + z_1 z_4 z_5 + z_2 z_3 z_4 +\\
& &  z_2 z_3 z_5 + z_1 z_2 z_6 + z_1 z_3 z_6 + z_2 z_4 z_6 + z_3 z_5 z_6 + z_4 z_5 z_6.
\end{array}
\]
Kushilevitz's function is defined as $h \diamond h \diamond \ldots \diamond h$.
\end{example}

Observe that the auxiliary function $h$ on 6 variables in
Kushilevitz's example has degree 3 and full sensitivity on the {\bf 0}
input. Let the function $f$ be obtained by composing $h$ with itself
$k$ times. It is defined on $n = 6^k$ variables and has full
sensitivity, block sensitivity, decision tree complexity and
ceritificate complexity. Degree of $f$ is $3^k = n^{\log_6 3}$.

The following function was constructed by Chakraborty~\cite{chak}. It
was discussed in Section~\ref{ProgressSec} of this note.
\begin{example}[Chakraborty's function]
\label{ChakrabortyEx}
Define an auxiliary function $h$ on $k^2$ variables by a regular
expression:
\[
h(z_{11}, \ldots, z_{kk}) = 1 \iff z \in 110^{k-2} (11111(0+1)^{k-5})^{k-2} 11111(0+1)^{k-8}111.
\]
Chakraborty's function $f$ on $n \ge k^2$ variables is defined as follows:
\[f(x_0, \ldots, x_{n-1}) = 1 \iff \big( \exists i \in [n] \big) \big( g(x_{i}, x_{(i+1)}, \ldots, x_{(i+k^2)}) = 1 \big), \]
where indices in the arguments of function $g$ are taken modulo $n$.
\end{example}

Chakraborty shows that for $n=k^3$ his function has sensitivity
$\Theta(n^{1/3})$~\cite{chak}, block sensitivity $\Theta(n^{2/3})$ and
certificate complexity $\Theta(n^{2/3})$~\cite{chakmaster}.



\noindent {\bf Acknowledgements:} We would like to thank Laci Babai
for reviewing preliminary versions of this note, giving helpful
comments, and introducing previous generations of the University of
Chicago students to the Sensitivity Conjecture, including David
Rubinstein, Sandy Kutin, and Sourav Chakraborty. Just as Laci seemed
to give up on popularizing the conjecture, Sasha Razborov rekindled
the flame. We would like to thank Sasha for introducing us to the
Sensitivity Conjecture.

\bibliographystyle{acm}
\bibliography{equiv.bib}

\end{document}